\documentclass[a4paper]{article}
\usepackage[british]{babel}
\usepackage{amsmath,amssymb,amsfonts,amsthm}
\usepackage{graphicx}
\usepackage{float}
\usepackage[perpage,symbol*]{footmisc}
\usepackage[plainpages=false,pdfpagelabels,hyperfootnotes=false]{hyperref}

\floatstyle{ruled}
\newfloat{algorithm}{htb}{alg}
\floatname{algorithm}{Algorithm}

\theoremstyle{plain}
\newtheorem{lemma}{Lemma}

\newtheorem{corollary}{Corollary}
\newtheorem{proposition}{Proposition}

\theoremstyle{definition}

\newtheorem{remark}{Remark}

\newtheorem{example}{Example}

\title{Scheduling unit processing time arc shutdown jobs to maximize network flow over time: complexity results}

\author{Natashia Boland \and Thomas Kalinowski \and Reena Kapoor \and Simranjit Kaur}
\date{}

\begin{document}

\maketitle

\begin{abstract}
We study the problem of scheduling maintenance on arcs of a capacitated network so as to maximize the total flow from a source node to a sink node over a set of time periods. Maintenance on an arc shuts down the arc for the duration of the period in which its maintenance is scheduled, making its capacity zero for that period. A set of arcs is designated to have maintenance during the planning period, which will require each to be shut down for exactly one time period. In general this problem is known to be NP-hard. Here we identify a number of characteristics that are relevant for the complexity of instance classes. In particular, we discuss instances with restrictions on the set of arcs that have maintenance to be scheduled; series parallel networks; capacities that are balanced, in the sense that the total capacity of arcs entering a (non-terminal) node equals the total capacity of arcs leaving the node; and identical capacities on all arcs.
\end{abstract}

\section*{Introduction}

Many real life systems can be viewed as a network with arc capacities, supporting the flow of a commodity. For example, transportation networks, or supply chains, may on occasion be viewed this way. We were motivated by a particular coal export supply chain~\cite{boland2010optimizinghvcc}, in which maximizing throughput is a key concern. Whilst this suggests a maximum flow model would be appropriate, in fact, the real network is not static: capacities change over time, and in particular, some arcs are shut down for maintenance at certain times. Often there is some flexibility in the time when maintenance jobs can be scheduled. Every maintenance schedule will incur some loss in the total throughput of the network. To obtain maximum throughput, it is important to select the schedule that leads to minimum loss of flow. For example consider the network in Figure~\ref{fig:example} with three nodes $\{s,v,t\}$, three arcs $\{a,b,c\}$ and given arc capacities.
\begin{figure}[htb]
  \centering
\includegraphics[width=.5\textwidth]{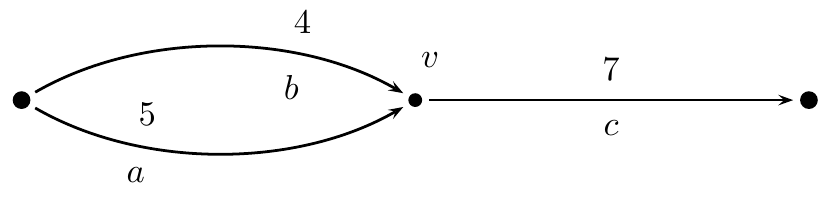}
\caption{Example network.}\label{fig:example}
\end{figure}
The total throughput possible in two time periods when no arcs are on maintenance is 14 units. Suppose that arc $a$ and $b$ have to go on maintenance for a unit period of processing time in a time horizon of two periods. The two possible schedules are either put both the arcs $a$ and $b$ on maintenance together in the first time period giving a total throughput of 7 units in two time periods or put the arc $a$ in first time period and arc $b$ in second time period giving the total throughput of 9 units in the two time periods. Clearly the second schedule is better than the first one as it is giving less loss of flow. This leads to a model in which arc maintenance jobs need to be scheduled so as to maximize the total flow in the network over time~\cite{boland2011optimisation,boland2012mixed,boland2012scheduling}.

In this paper we consider the case of this problem in which all maintenance jobs have unit processing time. The problem is defined over a network $N=(V,A,s,t,u)$ with node set $V$, arc set $A$, source $s\in V$, sink $t\in V$ and nonnegative integral capacity vector $u=(u_a)_{a\in A}$. Note that we permit parallel arcs, i.e. there may exist more than one arc in $A$ having the same start and end node, so $A$ is a multiset. By $\delta^-(v)$ and $\delta^+(v)$ we denote the set of incoming and outgoing arcs of node $v$, respectively. We consider this network over a set of $T$ time periods indexed by the set $[T]:=\{1,2,\ldots,T\}$, and our objective is to maximize the total flow from $s$ to $t$. In addition, we are given a subset $J\subseteq A$ of arcs that have to be shut down for exactly one time period in the time horizon. In other words, there is a set of maintenance jobs, one for each arc in $J$, each with unit processing time. Our optimization problem is to choose these outage time periods in such a way that the total flow from $s$ to $t$ is maximized. More formally, this can be written as a mixed binary program as follows:
\begin{align}
\max\ z = \sum_{i=1}^T&\left(\sum_{a\in\delta^+(s)}x_{ai}-\sum_{a\in\delta^-(s)}x_{ai}\right)  \label{eq:objective} \\
\text{s.t.}\qquad  x_{ai} &\leqslant u_a && a\in A\setminus J,\ i\in[T], \label{eq:cap_normal_arcs} \\
 x_{ai} &\leqslant u_ay_{ai} && a\in J,\ i\in[T], \label{eq:cap_job_arcs} \\
\sum_{i=1}^Ty_{ai} &= T-1 && a\in J, \label{eq:scheduling} \\
\sum_{a\in\delta^-(v)}x_{ai} &= \sum_{a\in\delta^+(v)}x_{ai} && v\in N\setminus\{s,t\},\ i\in[T], \label{eq:flow_conservation}\\
x_{ai} &\geqslant 0 && a\in A,\ i\in[T], \label{eq:nonnegative_flows}\\
y_{ai} &\in\{0,1\} && a\in J,\ i\in[T], \label{eq:binary_indicators}
\end{align}
where $x_{ai} \geqslant 0$ for $a \in A$ and $i \in [T]$ denotes the flow on arc $a$ in time period $i$, and $y_{ai} \in \{0,1\}$ for $a \in J$ and $i \in [T]$ indicates when the arc $a$ is {\em not} shut down for maintenance in time period $i$, i.e. $y_{ai}=0$ in the period $i$ in which the outage for arc $a$ is scheduled.

To the best of our knowledge, Boland et al.~\cite{boland2011optimisation,boland2012mixed,boland2012scheduling} initiated study on the problem with general processing times. In~\cite{boland2011optimisation,boland2012mixed}, the coal supply chain application, which has a number of additional side constraints, is modelled and solved using a rolling time horizon mixed integer programming approach. In~\cite{boland2012scheduling}, the complexity of the general problem is established, and four local search heuristics are developed and compared. We are not aware of any other studies on this problem. Several authors have studied dynamic network flows. For instance~\cite{ford1962flows} studied the problem of finding the maximum flow that can be sent from a source to a sink in $T$ time units, in a network with transit times on the arcs. Variations of the dynamic maximum flow problem with zero transit times are discussed in~\cite{fleischer2001universally},~\cite{hajek1984optimal}, and~\cite{hoppe1994polynomial}. None of these have a scheduling component. Machine scheduling problems have received a great deal of attention in the literature~\cite{pinedo2012scheduling}, but in the problem we study here, there is no underlying machine, and the association of jobs with network arcs and a maximum flow objective give it quite a different character. The closest work we can find is that in the recent paper of Tawarmalani and Li~\cite{tawarmalani2011multi}, which considers multiperiod maintenance scheduling over a network, in which the objective is based on multicommodity flows (with origin-destination demands, but without arc capacities), there is a limit on the number of arcs that can be shut down in any one period, and the network's structure is restricted to a tree. Complexity results are provided for linear networks, with a polynomial algorithm in the case of (nearly) uniform commodity demands, and a proof that the case of general demands is strongly NP-hard. Integer programming models are also considered, and polyhedral analysis carried out. The lack of previous attention to the trade-off between maintenance scheduling and network flow reduction in the literature is also noted in~\cite{tawarmalani2011multi}.

Our key contribution in this paper is an analysis of how the complexity of the problem depends on important characteristics: (i) the case that the set of arcs with a job contains a minimum cut of the network, (ii) {\em balanced} networks, in which the capacity into and out of each (non-terminal, i.e. transhipment) node is equal, (iii) networks that are series-parallel, (iv) the number of time periods is treated as a fixed parameter, and (v) the case that all arcs have the same capacity. We show for case (i) that it is optimal to schedule all jobs in the same time period, and that this is also true if the network is both balanced and series-parallel. However if the network is balanced but not necessarily series-parallel, then the problem is strongly NP-hard. We provide an approximation ratio for scheduling all jobs in the same time period in the general case, which shows this is asymptotically optimal as $T$ approaches infinity. For case (iv), we show that even if $T=2$ and the network contains only a single transhipment node, the problem is weakly NP-hard, and we give an algorithm for series-parallel networks that has pseudopolynomial complexity for $T$ fixed (but is exponential in $T)$. In case (v), if all arcs have the same capacity, we prove that the problem can be reduced to a maximum flow problem and $T$ additional linear programs, and hence can be solved in polynomial time. In this case it is not necessarily optimal to schedule all jobs at the same time.

The paper is organized as follows. Section~\ref{sec:single_node} contains a discussion of cases of the network with a single transhipment node. In Section~\ref{sec:general_networks} we explore general networks, and in Section \ref{sec:unit_cap} we consider the case that all arcs have the same capacity. Finally, in Section~\ref{sec:future} we suggest some future directions for study of this problem.

\section{Networks with single transhipment node}\label{sec:single_node}

The problem in general is NP-hard~\cite{boland2012scheduling}. In this proof, the reduction gave rise to a network with a single transhipment node, which was not balanced, and in which the set of arcs with associated jobs did not contain a minimum cut. This left open the complexity of the cases that all arcs in a minimum cut have an associated outage, or the network is balanced. This section gives a result that describes a class of networks with single transhipment node that covers the above-mentioned cases and is easy to resolve. Consider a network having only one transhipment node, say $v$. Let 
\begin{align*}
  C_1^- &= \sum_{a\in\delta^-(v)} u_a, & C_1^+ &= \sum_{a\in\delta^+(v)} u_a, \\
C_2^- &= \sum_{a\in\delta^-(v)\setminus J} u_a, & C_2^+ &= \sum_{a\in\delta^+(v)\setminus J} u_a.
\end{align*}
If all jobs are scheduled at the same time, say in time period 1, we obtain a total throughput of
\[\min\left\{C_2^-,\,C_2^+\right\}+(T-1)\min\left\{C_1^-,\,C_1^+\right\}.\]
On the other hand, using $\displaystyle\sum_{a\in\delta^-(v)\cap J}u_a= C_1^--C_2^-$ and $\displaystyle\sum_{a\in\delta^+(v)\cap J}u_a = C_1^+-C_2^+$ we obtain an upper bound of
\begin{multline*}
\min\left\{TC_2^-+(T-1)(C_1^--C_2^-),\,TC_2^++(T-1)(C_1^+-C_2^+)\right\}\\
=\min\left\{C_2^-+(T-1)C_1^-,\,C_2^++(T-1)C_1^+\right\}.
\end{multline*}
If (i) $C_1^-\leqslant C_1^+$ and $C_2^-\leqslant C_2^+$ or (ii) $C_1^+\leqslant C_1^-$ and $C_2^+\leqslant C_2^-$ then the upper bound equals the lower bound, and this proves the following sufficient optimality condition.
\begin{proposition}
If (i) $C_1^-\leqslant C_1^+$ and $C_2^-\leqslant C_2^+$ or (ii) $C_1^+\leqslant C_1^-$ and $C_2^+\leqslant C_2^-$, then it is optimal to schedule all jobs at the same time.\qedhere
\end{proposition}

As a simple consequence we note that in the following situations it is optimal to schedule all jobs at the same time:
\begin{itemize}
\item $C_1^-=C_1^+$, so the network is balanced, or
\item $C_1^-\leqslant C_1^+$ and $J\supseteq\delta^-(v)$, or
\item $C_1^+\leqslant C_1^-$ and $J\supseteq\delta^+(v)$.
\end{itemize}
For a time horizon of two time periods the problem asks for a partition of the job set $J=J_1\cup J_2$ into two parts such that the total flow is maximized, i.e. we want to find
\[\max\limits_{J_1\cup J_2=J}\left[\min\left\{\sum_{a\in\delta^-(v)\setminus J_1}u_a,\,\sum_{a\in\delta^+(v)\setminus J_1}u_a\right\}+\min\left\{\sum_{a\in\delta^-(v)\setminus J_2}u_a,\,\sum_{a\in\delta^+(v)\setminus J_2}u_a\right\}\right].\]
The following proposition shows that it is NP-hard to decide if the trivial partition $J_1=J$ and $J_2=\varnothing$ is optimal.
\begin{proposition}\label{prop:single_node}
 For a network with one transhipment node and a time horizon of two periods it is NP-hard to decide if it is optimal to schedule all jobs at time 1.
\end{proposition}
\begin{proof}
Reduction from \textsc{Partition} (see~\cite{garey1979computers}). An instance is given by a set $D=\{d_1,\ldots,d_m\}$ of positive integers with $\sum_{i=1}^m d_i=2B$, and the problem is to decide if there is a partition $D=D_1\cup D_2$ such that $\sum_{d\in D_1} d=\sum_{d\in D_2} d = B$. We consider the network shown in Figure~\ref{fig:partition_network} where every arc except the bold arc from $v$ to $t$ has an associated job.
\begin{figure}[htb]
  \centering
\includegraphics[width=.6\textwidth]{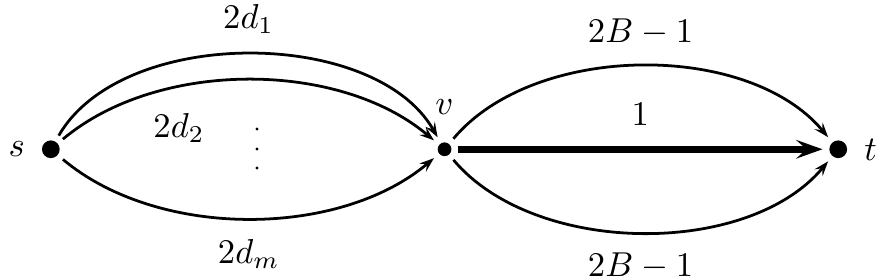}
\caption{The network for the reduction from \textsc{Partition}. Arcs are labeled with capacities.}\label{fig:partition_network}
\end{figure}
Scheduling all jobs at time 1 gives a total flow of $4B-1$. A total flow of $4B$ is possible if and only if there is a flow of $2B$ in each time period, and this is equivalent to a positive solution for the \textsc{Partition} instance.
\end{proof}

The reduction from \textsc{Partition} suggests the use of dynamic programming to obtain a pseudopolynomial algorithm for the single node problem. This is indeed possible, and in fact can be done more generally for series-parallel networks. This more general approach is presented in the next section (see Corollary~\ref{cor:series_parallel_fixed_time}).

\section{General Networks}\label{sec:general_networks}
In this section we explore complexity issues for networks with more than one transhipment node and also discuss some of its tractable subclasses. We start with a lemma generalizing the upper bound in the single node case.
\begin{lemma}\label{lem:upper_bound}
Let $S\subseteq A$ be any $s$-$t$ cut in the network. The objective value for problem (\ref{eq:objective})~--~(\ref{eq:binary_indicators}) is bounded above by
\[T\sum_{a\in S\setminus J} u_a+(T-1)\sum_{a\in S\cap J} u_a.\]
\end{lemma}
\begin{proof}
Since $S$ is a cut, the total flow over the whole time horizon is bounded above by
\[  \sum_{i=1}^T\sum_{a\in S} x_{ai}=\sum_{a\in S}\sum_{i=1}^T x_{ai}
    = \sum_{a\in S\setminus J} \sum_{i=1}^T x_{ai} + \sum_{a\in S\cap J} \sum_{i=1}^T x_{ai}
    \leqslant \sum_{a\in S\setminus J} Tu_a + \sum_{a\in S\cap J}
    (T-1) u_a,
    \]
by the combination of (\ref{eq:cap_normal_arcs})~---~(\ref{eq:scheduling}). The result follows. \qedhere
\end{proof}

As an immediate consequence we obtain that the problem is tractable when the set of arcs that have to undergo maintenance contains a minimum cut.
\begin{proposition}\label{prop:all_arcs}
If $J$ contains a minimum cut $S$ of the network then it is optimal to schedule all jobs at the same time.
\end{proposition}
\begin{proof}
 Since $S$ is a minimum cut, the maximum flow in any period in which no maintenance is scheduled is $\sum_{a\in S} u_a$, so scheduling all jobs at time 1 gives a total flow of $(T-1)\sum_{a\in S} u_a$, which achieves the upper bound from Lemma~\ref{lem:upper_bound}.
\end{proof}

In Section \ref{sec:single_node}, we showed that the case of single-node networks with balanced capacities is easy. The following theorem shows that the balanced property alone is not enough. 
\begin{proposition}\label{prop:balanced_is_hard}
  The problem is strongly NP-hard for balanced networks.
\end{proposition}
\begin{proof}
Reduction from \textsc{3-Partition} (see~\cite{garey1979computers}). A \textsc{3-Partition} instance is given by an integer $B$ and a set $\{d_1,\ldots,d_{3m}\}$ of integers with $B/4<d_i<B/2$ for all $i$ and $\sum_{i=1}^{3m} d_i=mB$. The problem is to decide if there is a partition of the set $\{d_1,\ldots,d_{3m}\}$ into $m$ triples such that the sum of each triple equals $B$. Consider the network shown in Figure \ref{fig:3partition_network}, where the arc labels indicate capacities, and the bold arcs don't have jobs associated with them. Also let the time horizon be $T=m$.
\begin{figure}[htb]
  \centering
\includegraphics[width=.8\textwidth]{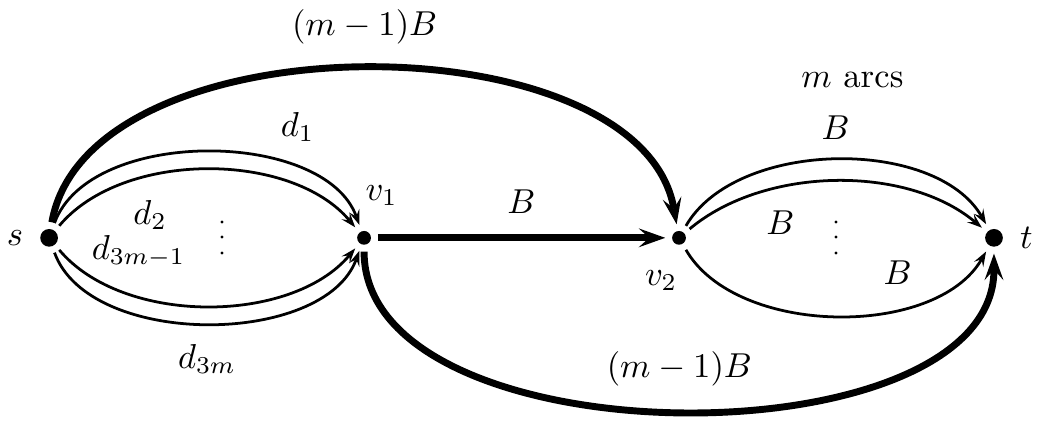}
\caption{The network for the reduction from \textsc{3-Partition}.}\label{fig:3partition_network}
\end{figure}
By Lemma \ref{lem:upper_bound} applied to the cut $(\{s\},\{v_1, v_2,t\})$, the total flow is bounded by 
\[T(m-1)B+(T-1)\sum_{i=1}^{3m} d_i=2m(m-1)B.\]
To achieve the bound $2m(m-1)B$ the arc $(s,v_2)$ is at capacity in every time period. This implies that we have to schedule exactly one job on the arcs between $v_2$ and $t$ in each time period. Now flow conservation in node $v_2$ implies that the flow on the arc $(v_1,v_2)$ is zero in every time period. Considering the cut $(\{s,v_1,v_2\},\{t\})$ the bound $2m(m-1)B$ can be achieved only if the arc $(v_1,t)$ is at capacity in every time period. Using flow conservation in node $v_1$ we can now conclude that in order to achieve the bound $2m(m-1)B$ it is necessary and sufficient to send in each time period $(m-1)B$ units of flow from $s$ to $v_1$, and this can be done if and only if the answer for the \textsc{3-Partition} instance is YES.
\end{proof}

As already mentioned, not all instances of the general balanced network are hard. The single-node variant is easy and is in fact a special case of a series-parallel network. Note that the network constructed in the above NP-hardness proof is not series-parallel. We show below (Proposition \ref{prop:series_parallel_balanced}) that indeed the case of series-parallel balanced networks is easy. However we first make precise our definition of series-parallel. Throughout this paper, by \emph{series-parallel network} we mean a \emph{two-terminal series-parallel network}: a network that has a single source and single sink and is constructed by a sequence of series and parallel compositions starting from single arcs. For two networks $N_1$ and $N_2$  the  \emph{parallel composition} of $N_1$ and $N_2$ is obtained by identifying the source node $s_1$ and sink node $t_1$ of $N_1$ with the source node $s_2$ and sink node $t_2$ of $N_2$, respectively. The \emph{series composition} of $N_1$ and $N_2$ is obtained by identifying the sink node $t_1$ of $N_1$ with the source node $s_2$ of $N_2$. We denote these compositions by $N_1\oplus_P N_2$ and $N_1\oplus_S N_2$, respectively. The next proposition shows that series-parallel balanced networks are tractable.

\begin{proposition}\label{prop:series_parallel_balanced}
  If the network is series-parallel and balanced then it is optimal to schedule all jobs at the same time.
\end{proposition}

\begin{proof}
For a network $N=(V,A,s,t,u)$ and a subset $J\subseteq A$ let $F_{N,J}$ denote the maximum flow value in the network $N=(V,A\setminus J,s,t,u\mid_{A\setminus J})$. The statement that it is optimal to schedule all jobs at the same time is equivalent to
\[F_{N,J\cup J'}+F_{N,\varnothing} \geqslant F_{N,J} + F_{N,J'}\]
for all $J,J'\subseteq A$ (see~\cite{boland2012scheduling}). We prove the proposition by induction on the structure of the graph. The claim holds for the base case of a single arc. So assume that $N$ is a series-parallel network that is not a single arc. Then $N=N_1\oplus_P N_2$ or $N=N_1\oplus_S N_2$ for some smaller networks $N_i=(V_i,A_i,s_i,t_i,u_i)$ ($i\in\{1,2\}$), and by induction
\[F_{N_i,J_i\cup J_i'}+F_{N_i,\varnothing} \geqslant F_{N_i,J_i} + F_{N,J_i'}\]
for all $J_i,J'_i\subseteq A_i$. Now let $J,J'\subseteq A=A_1\cup A_2$ be arbitrary and put $J_i=J\cap A_i$ and $J'_i=J'\cap A_i$ for $i\in\{1,2\}$.
\begin{description}
\item[Case 1.] $N=N_1\oplus_P N_2$. Then
  \begin{multline*}
    F_{N,J\cup J'}+F_{N,\varnothing} = F_{N_1,J_1\cup J'_1}+F_{N_1,\varnothing} +
    F_{N_2,J_2\cup J'_2}+F_{N_2,\varnothing} \\ \geqslant  F_{N_1,J_1}+F_{N_1,J'_1} + F_{N_2,J_2}+F_{N_2,J'_2}  = F_{N,J} + F_{N,J'}.
  \end{multline*}
\item[Case 2.] $N=N_1\oplus_S N_2$. By the assumption that $N$ is balanced, we have $F_{N,\varnothing}=F_{N_1,\varnothing}=F_{N_2,\varnothing}$, and we denote this common value by $F$. Now
  \begin{multline*}
    F_{N,J\cup J'}+F=\min\{F_{N_1,J_1\cup J'_1},\,F_{N_2,J_2\cup J'_2}\}+F \geqslant \min\{F_{N_1,J_1}+F_{N_1,J'_1},\,F_{N_2,J_2}+F_{N_2,J'_2}\}\\
\geqslant\min\{F_{N_1,J_1},\,F_{N_2,J_2}\}+\min\{F_{N_1,J'_1},\,F_{N_2,J'_2}\} = F_{N,J} + F_{N,J'}.\qedhere
  \end{multline*}
\end{description}
\end{proof}

Since scheduling all jobs in the same period seems to be optimal in some cases, we now ask how well it performs as an approximation algorithm in the general case. 
\begin{proposition}
\label{prop:approx} Scheduling all jobs in the same period gives an approximation ratio no less than $\frac{(T-1)}{T}$.
\end{proposition}
\begin{proof} Let $z^*$ denote the optimal value and $\tilde{z}$ denote the throughput obtained by scheduling all arcs in the same period. Clearly $\tilde{z} \geqslant (T-1)F$ and $z^* \leqslant TF$, so
\[\frac{\tilde{z}}{z^*} \geqslant \frac{T-1}{T}.\qedhere\]
\end{proof}
Thus scheduling all jobs in the same period is asymptotically optimal in the sense that the approximation ratio approaches 1 as $T$ tends to infinity. In general the analysis in the proof of Proposition \ref{prop:approx} is tight as can be seen by considering the network in Figure \ref{fig:worstapprox} where all arcs have unit capacity and the set $J$ of arcs with a job is the set of the two arcs from $s$ to $v$. Then scheduling both outages at the same time yields a total throughput of $T-1$ while for $T\geqslant 2$ the outages can be scheduled in different time periods which yields a total throughput of $T$, and the approximation ratio in this case is $(T-1)/T$.
 \begin{figure}[htb]
  \centering
\includegraphics[width=.45\textwidth]{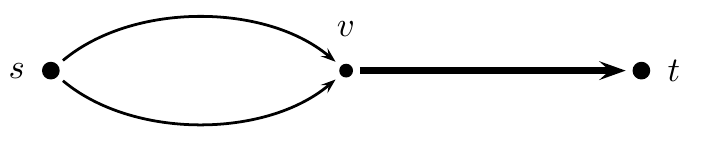}
\caption{A network where the bound of Proposition~\ref{prop:approx} is tight.}\label{fig:worstapprox}
\end{figure}
For certain instances the analysis of the approximation ratio can be slightly improved. For this let
\begin{align*}
L &= (T-1)F+\min_{S\in\mathcal S}\sum_{a\in S\setminus J}u_a,&
U &= \min\limits_{S\in\mathcal S}\left(T\sum_{a\in S}u_a-\sum_{a\in S\cap J}u_a\right)
\end{align*}
where the minima are over the the set $\mathcal S$ of all $s$-$t$-cuts in the network. Clearly, $L$ is the objective value for scheduling all jobs in the same time period, and from Lemma \ref{lem:upper_bound} it follows that $U$ is an upper bound for the optimal objective value. Thus $L/U$ is a lower bound for the approximation ratio, and since $L\geqslant (T-1)F$ and $U\leqslant TF$ this is at least as good as the bound from Proposition \ref{prop:approx}. Note that this generalizes Proposition \ref{prop:all_arcs}: if $J$ contains a min cut $S$ then both of the minima in the definitions of $L$ and $U$ are obtained for $S$, and we get $L=U=(T-1)F$, the approximation ratio is 1, in other words it is optimal to schedule all jobs in the same time period. 

Next we present an algorithm for general series-parallel networks, which for the instance used in the proof of Proposition~\ref{prop:single_node} coincides with the well known dynamic programming algorithm for \textsc{Partition}. With feasible values for the binary variables $y_{ai}$ ($a\in J$, $i\in[T]$) we can associate a vector $z^y=(z^y_i)_{i=1,\ldots,T}$ where $z^{y}_i$ denotes the maximum flow in the network with arc set $A\setminus\{a\ :\ y_{ai}=0\}$. By symmetry we may assume that $z^{y}_1\geqslant z^{y}_2\geqslant\cdots\geqslant z^{y}_T$. Our algorithm exploits the fact that many different maintenance schedules $y$ may give rise to the same vector $z^y$ to gain efficiency over naive enumeration of schedules. The algorithm computes the possible vectors $z$ for subnetworks of the network $N$, starting from single arcs. To do this we use \emph{sp-trees} which encode the construction of series-parallel networks. An sp-tree for a series-parallel network $N$ is a full binary tree in which the leaves correspond to the arcs of $N$, any internal node corresponds to the composition of its two child nodes, and the type of composition (series or parallel) is indicated by a node label (`S' or `P', respectively). Figure~\ref{fig:sp_network} shows a network and the corresponding sp-tree.
\begin{figure}[htb]
\begin{minipage}{.48\linewidth}
  \centering
\includegraphics[width=.8\textwidth]{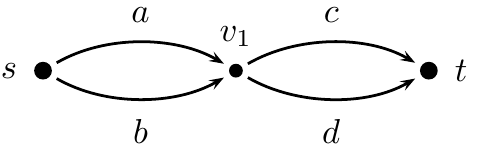}
\end{minipage}\hfill
\begin{minipage}{.48\linewidth}
  \centering
\includegraphics[width=.7\textwidth]{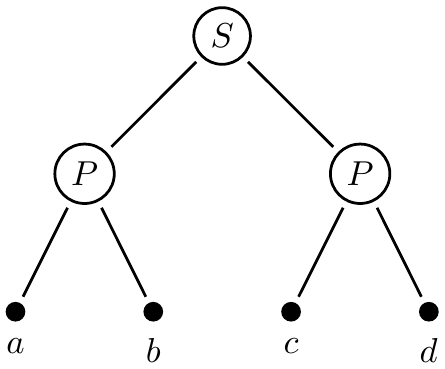}
\end{minipage}
\caption{A series-parallel network and the corresponding sp-tree.}\label{fig:sp_network}
\end{figure}
Recognition of series-parallel networks and construction of an sp-tree can be done in linear time~\cite{valdes1979recognition}. So assume we are given the sp-tree with node set $\mathcal V=\mathcal L\cup\mathcal W$, where $\mathcal L$ is the set of leaves and $\mathcal W$ is the set of internal nodes. The set $\mathcal W$ is partitioned into level sets $\mathcal W_i$ where $\mathcal W_i$ is the set of internal nodes at distance $i$ from the root. Let $d$ be the largest index such that $\mathcal W_d\neq\varnothing$. The lists of possible maximum flow vectors $z$ are initialized at the leaves by assigning a list with a single element to the leaf corresponding to arc $a$. The unique element in this list is $(u_a,u_a,\ldots,u_a,0)$ if $a\in J$ and $(u_a,u_a,\ldots,u_a,u_a)$ if $a\not\in J$. Then the lists for the internal nodes are computed going up in the tree as described in Algorithm~\ref{alg:sp_networks}. The list generated for each node $v\in \mathcal V$ is denoted by $L_v$.
\begin{algorithm}
  \caption{Maximizing total throughput for series-parallel networks}\label{alg:sp_networks}
  \begin{tabbing}
    .....\=.....\=.....\=.....\=.....\=................... \kill \\
\textbf{for} $v\in\mathcal L$ \textbf{do} \\
\> Let $a\in A$ be the arc corresponding to $v$\\
\> \textbf{if} $a\in J$ \textbf{then} $L_v\leftarrow\{(u_a,u_a,\ldots,u_a,0)\}$
\textbf{else} $L_v\leftarrow\{(u_a,u_a,\ldots,u_a,u_a)\}$\\
\textbf{for} $i=d,d-1,\ldots,0$ \textbf{do}\\
\> \textbf{for} $v\in\mathcal W_i$ \textbf{do}\\
\> \> $L_v\leftarrow\{\}$\qquad /* initialize empty list*/\\
\> \> Let $u$ and $w$ be the child nodes of $v$\\
\> \> \textbf{for} each $z\in L_u$, $z'\in L_w$ and $\pi$ a permutation of $\{1,2\ldots,T\}$ \textbf{do}\\
\> \> \> \textbf{if} $v$ is a parallel composition node \textbf{then}\\
\> \> \> \> \textbf{for} $i\in[T]$ \textbf{do} $z''_i=z_i+z'_{\pi(i)}$\\
\> \> \> \textbf{else} /* $v$ is a series composition node */\\
\> \> \> \> \textbf{for} $i\in[T]$ \textbf{do} $z''_i=\min\{z_i,\,z'_{\pi(i)}\}$\\
\> \> \> sort the components of $z''$ in non-increasing order\\
\> \> \> \textbf{if} $z''\not\in L_v$ \textbf{then} add $z''$ to $L_v$\\
Let $v$ be the root node of the sp-tree and return
$\max\limits_{z\in L_v}\sum\limits_{i=1}^Tz_i$
  \end{tabbing}
\end{algorithm}
This algorithm returns the maximum total throughput, and it is easy to see how to keep track of corresponding schedules for all the elements of the lists in the internal nodes.
\begin{example}
Suppose for the network in Figure~\ref{fig:sp_network} the capacities are $u_a=4$, $u_b=1$, $u_c=u_d=2$, the set of arcs with a job is $J=\{a,b,c\}$, and the time horizon is $T=3$. Figure~\ref{fig:algorithm_run} illustrates how the lists for the internal nodes are computed.
\begin{figure}[htb]
  \centering
\includegraphics[width=.63\textwidth]{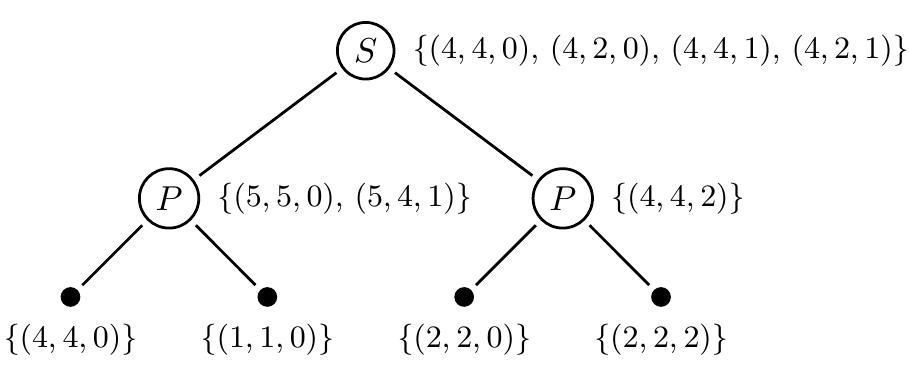}
  \caption{Computation of the possible maximum flow vectors.}
  \label{fig:algorithm_run}
\end{figure}
The optimal vector in the root node is $(4,4,1)$ giving a total throughput of 9, and this can be obtained by scheduling the job on arc $b$ for the second time period and the jobs for arcs $a$ and $c$ for the third time period.
\end{example}
Bounding the runtime of Algorithm \ref{alg:sp_networks} we obtain the following complexity result.
\begin{proposition}\label{prop:series_parallel_complexity}
For series-parallel networks with $m$ arcs the problem can be solved in time 
\[O\left(mT^{T+3/2}e^{-T}\log(T)(mB+1)^{2T}\right),\] 
where $B$ is an upper bound for the capacities.
\end{proposition}
\begin{proof}
The entries of the vectors in the lists at the internal nodes are bounded by $mB$, hence every list can contain at most $(mB+1)^T$ elements. Thus the loop over $(z,z')\in L_u\times L_w$ and permutations $\pi$ is over at most $T!(mB+1)^{2T}$ elements. Inside this loop is another one giving an additional factor $T$, and a sorting operation which is at most a factor of $T \log(T)$. With the use of hash tables, checking $z''$ is not already in the list prior to insertion will not worsen the complexity of operations inside this loop, which is thus $T\log(T)$. In total there are $m-1$ internal nodes, thus the runtime is $O(T\log(T)T!(mB+1)^{2T}(m-1))$ from which the result follows when $T!$ is bounded using Stirling's formula.
\end{proof}

We add two remarks on an efficient implementation of Algorithm \ref{alg:sp_networks}.
\begin{enumerate}
\item Any vector $z$ that is dominated by another vector $z'$ in the list, meaning that $z_i\leqslant z'_i$ for all $i\in\{1,2\ldots,T\}$, can be removed immediately.
\item In the loop over $(z,z')\in L_u\times L_w$ and permutations $\pi$ it is necessary to loop over all permutations only if the entries of the vectors $z$ and $z'$ are pairwise distinct. An efficient implementation detects the occurrence of multiple entries and restricts the range of the considered permutations accordingly.
\end{enumerate}

Note that the first remark slightly sharpens the analysis in the proof of Proposition~\ref{prop:series_parallel_complexity}. Since no two vectors $z$ in any list will coincide in all but one entry, the lengths of the lists are bounded by $(mB+1)^{T-1}$. Taking this into account, we obtain the following run-time bound for a fixed time horizon. 
\begin{corollary}\label{cor:series_parallel_fixed_time}
For series-parallel networks and fixed time horizon $T$ the problem can be solved in time $O\left(m^{2T-1}B^{2T-2}\right)$.  
\end{corollary}

\section{Networks with all arcs having unit capacity}\label{sec:unit_cap}
In this section we study the case that the capacity of every arc equals 1. We can aggregate all time periods and solve a standard max flow problem to get an upper bound. The max flow problem is
\begin{align}
  \max\ \sum_{a\in\delta^+(s)}X_{a}&-\sum_{a\in\delta^-(s)}X_{a} \label{eq:objective_agg} \\
\text{s.t.} \sum_{a\in\delta^+(v)}X_{a} &= \sum_{a\in\delta^-(v)}X_{a} && v\in N\setminus\{s,t\},\label{eq:flow_conservation_agg}\\
X_{a} &\leqslant T && a\in A\setminus J, \label{eq:capacity_1_agg} \\
X_{a} &\leqslant T-1 && a\in J, \label{eq:capacity_2_agg} \\
X_{a} &\geqslant 0 && a\in A.\label{eq:nonnegativity_agg}
\end{align}
We will show that this upper bound is actually tight. This follows
by induction once we can find a max flow $X^*$ and cover all the
arcs carrying flow $T$ by a collection of arc disjoint
$s$-$t$-paths. Given any max flow we can reduce the flow along any
cycles carrying flow, and we can remove arcs with zero flow. So in
order to prove that the upper bound is tight it is sufficient to
prove the following result.
\begin{proposition}\label{prop:extract_unit_flows}
Let $(V,A,s,t)$ be an acyclic network with source $s\in V$ and sink $t\in V$, and suppose $X^*:A\to[T]$ satisfies the flow conservation constraints
\[\sum_{a\in\delta^-(v)}X^*_a=\sum_{a\in\delta^+(v)}X^*_a\qquad \text{for all } v\in V\setminus\{s,t\}.\]
Then there is a collection $\mathcal P$ of arc-disjoint $s$-$t$-paths such that $A^*\subseteq\bigcup_{P\in\mathcal P}P$, where $A^*=\{a\in A\ :\ X^*_a=T\}$ is the set of arcs carrying flow $T$.
\end{proposition}
\begin{proof}
We consider the following binary program, in which $(\xi_a)_{a\in
A}$ induces a set of arc-disjoint $s$-$t$-paths:
\begin{align}
  \max\ \sum_{a\in A^*}\xi_a &\label{eq:obj_aux} \\
\text{s.t.}\quad \sum_{a\in\delta^+(v)}\xi_a-\sum_{a\in\delta^-(v)}\xi_a &=0  && v\in V\setminus\{s,t\}, \label{eq:flow_con_aux} \\
\xi_a&\in\{0,1\} && a\in A.\label{eq:domain_aux}
\end{align}
We have to prove that the optimal objective value for the problem
(\ref{eq:obj_aux})~---~(\ref{eq:domain_aux}) is $\lvert A^*\rvert$.
The flow conservation constraints (\ref{eq:flow_con_aux}) form a
network matrix, hence we do not lose anything by relaxing
integrality, i.e. we can replace (\ref{eq:domain_aux}) by
$0\leqslant \xi_a\leqslant 1$ for all $a\in A$. The dual problem can
be written in the form
\begin{align}
  \min\ \sum_{a\in A}\eta_a& \label{obj_dual_aux} \\
\text{s.t.}\quad\pi_{v}-\pi_w+\eta_a  &\geqslant 0 && a=(v,w)\in A\setminus A^*, \label{eq:pot_diff_1} \\
 \pi_{v}-\pi_w+\eta_a  &\geqslant 1 && a=(v,w)\in A^*, \label{eq:pot_diff_2} \\
\pi_s=\pi_t&=0, \label{eq:boundaries} \\
\eta_a &\geqslant 0 && a\in A. \label{eq:nonnegativity}
\end{align}
A feasible solution with objective value $\lvert A^*\rvert$ is given by $\pi_v=0$ for all $v\in V$, $\eta_a=0$ for $a\in A\setminus A^*$, and $\eta_a=1$ for $a\in A^*$. In order to prove our claim we have to check that $\lvert A^*\rvert$ is a lower bound, i.e. that $\sum_{a\in A}\eta_a\geqslant\lvert A^*\rvert$ for every feasible solution. To see this let $\mathcal P'$ be any decomposition of the flow $X^*$ into paths, that is a collection of $s$-$t$-paths such that every arc $a$ is contained in exactly $X^*_a$ paths $P\in\mathcal P'$. Adding up constraints (\ref{eq:pot_diff_1}) and (\ref{eq:pot_diff_2}) over the arcs of any path $P\in\mathcal P'$, we obtain $\sum\limits_{a\in P}\eta_a\geqslant\lvert P\cap A^*\rvert$, hence
\[\sum_{a\in A} X^*_a\eta_a=\sum_{P\in\mathcal P'}\sum_{a\in P}\eta_a
  \geqslant\sum_{P\in\mathcal  P'}\lvert P\cap A^*\rvert=T\lvert A^*\rvert.\]
Finally, using $X^*_a\leqslant T$ for all $a\in A$,
\[\sum_{a\in A}\eta_a\geqslant\sum_{a\in A}\frac{X^*_a}{T}\eta_a\geqslant\lvert A^*\rvert.\qedhere\]
\end{proof}
To summarize, if $u_a=1$ for all $a\in A$ then the problem can be reduced to solving the max flow problem (\ref{eq:objective_agg})~--~(\ref{eq:nonnegativity_agg}) followed by $T$ instances of (the linear relaxation of) the problem (\ref{eq:obj_aux})~--~(\ref{eq:domain_aux}). Consequently, these instances can be solved in time polynomial in the size of the network and the time horizon $T$.
\begin{remark}
It is straightforward to generalize the result of this section to the problem where every arc can have several unit processing time jobs which must not overlap. The only necessary modification in this case is to replace the right-hand side of constraint (\ref{eq:capacity_2_agg}) by $T-m_a$ where $m_a$ is the number of jobs that have to be scheduled on arc $a$.
\end{remark}

\section{Future Work}\label{sec:future}

It would be interesting to find a more combinatorial proof of Proposition \ref{prop:extract_unit_flows}, rather than resorting to solution of linear programs. A combinatorial proof may suggest combinatorial algorithms for constructing the arc-disjoint paths that cover all arcs with flow $T$. Practical algorithms for more general problems are also of interest.

\end{document}